\documentclass[a4,12pt]{article}%
\usepackage[cmex10]{amsmath}
\usepackage{amsmath,amsthm,amssymb}
\usepackage{epic, eepic}
\usepackage{graphicx}
\usepackage{amsfonts}
\usepackage{amssymb}
\usepackage{times}
\usepackage{color}
\usepackage{url}
\providecommand{\U}[1]{\protect\rule{.1in}{.1in}}
\newtheorem{theorem}{Theorem}[section]

\newtheorem{definition}{Definition}[section]
\newtheorem{example}{Example}[section]

\newtheorem{idea memo}{Idea Memo}[section]

\newtheorem{proposition}{Proposition}[section]

\newcommand{\beq}{\begin{equation}}
\newcommand{\eeq}{\end{equation}}
  \newcommand{\beql}[1]{\begin{equation}\label{eq:#1}}
  \newcommand{\beqa}{\begin{eqnarray}}
  \newcommand{\eeqa}{\end{eqnarray}}
  \newcommand{\beqas}{\begin{eqnarray*}}
  \newcommand{\eeqas}{\end{eqnarray*}}
  \newcommand{\ep}{\varepsilon}

\newcommand{\eq}[1]{(\ref{eq:#1})}
\begin{document}

\title{Universally valid uncertainty relations in general quantum systems}
\author{Kazuya Okamura\thanks{e-mail: okamura@math.cm.is.nagoya-u.ac.jp}$\;$
 and Masanao Ozawa\thanks{e-mail: ozawa@is.nagoya-u.ac.jp}
\vspace{2mm}\\
Graduate School of Informatics, Nagoya University,\\
Chikusa-ku, Nagoya 464-8601, Japan }
\date{}
\maketitle
\begin{abstract}
We study universally valid uncertainty relations in general quantum systems described by
general $\sigma$-finite von Neumann algebras to foster
developing quantitative analysis in quantum systems with infinite degrees of freedom
such as quantum fields. 
We obtain the most stringent measurement-disturbance relation ever, 
applicable to systems with infinite degrees of freedom, by refining the proofs 
given by Branciard and one of the authors (MO) for systems with
finite degrees of freedom. 
In our proof the theory of the standard form of von Neumann algebras plays 
a crucial role, incorporating with measurement theory for local quantum systems
recently developed by the authors.
\end{abstract}

\section{Introduction}\label{1}
Mathematical formalism of quantum theory introduced non-commutativity of 
observables. Hei\-sen\-berg's uncertainty relation elucidates an operational meaning of the non-commutativity 
as a limitation to the simultaneous measurability of a pair of observables.
In 1927, using the famous $\gamma$-ray microscope thought experiment, 
Heisenberg  \cite{Hei27} claimed that canonically conjugate observables $Q, P$ 
can be measured simultaneously only with the relation
\beql{Hei27}
\ep(Q)\ep(P)\ge\frac{\hbar}{2}
\eeq
for the ``mean errors''  $\ep(Q), \ep(P)$.  
However, his formal derivation of this relation from the well-established relation
\begin{equation}
\sigma(Q)\sigma(P)\ge\frac{\hbar}{2} \label{Ken27}
\end{equation}
for the standard deviations  $\sigma(Q), \sigma(P)$ due to Heisenberg \cite{Hei27} and Kennard \cite{Ken27}
needs an additional assumption such as a quantitative version of the repeatability hypothesis \cite{15A2}.
Although the repeatability hypothesis was commonly accepted at that time 
(cf.~Schr\"{o}dinger \cite[Section 8]{Sch35}, Dirac \cite[p.~36]{Dir58}, and von Neumann \cite[p.~335]{vN32E}),
this hypothesis has been completely abandoned in the modern quantum mechanics \cite{DL70},
in which quantum measurements are generally described by completely positive instruments \cite{84QC}. 
In such a general description of quantum measurements, Heisenberg's relation \eq{Hei27} 
loses its universal validity \cite{88MS,02KB5E}.
An alternative relation universally valid for arbitrary measurements, 
arbitrary pairs of observables, and arbitrary states was derived only recently by 
one of the authors \cite{03HUR,03UVR,03UPQ,04URJ,04URN,05UUP},
 and has recently received considerable attention. 
The validity of the above new relation, as well as a stronger version of this relation 
\cite{Bra13,Bra14,14EDR}, were experimentally tested with neutrons 
 \cite{12EDU,13VHE,16A3}
and with photons
 \cite{RDMHSS12,13EVR,WHPWP13,RBBFBW14,14ETE}.
Other approaches generalizing Heisenberg's original relation
can be found, for example, in 
\cite{App98c,Hal04,Wer04,BLW13,BLW14PRA,BLW14JMP,%
BLW14RMP,LYFO14},
apart from the information theoretical approach  \cite{14NDQ,15A1}.

All the above studies of universally valid uncertainty relations have been restricted to 
quantum systems with finite degrees of freedom. 
In the present paper, we show universally valid uncertainty relations for quantum systems 
described by general $\sigma$-finite von Neumann algebras including systems 
with infinite degrees of freedom, based on quantum measurement theory 
for systems with infinite degrees of freedom recently established 
by the present authors \cite{16A1}.
The authors aim to clarify the mathematical essence of derivations of 
universally valid uncertainty relations and to give a simpler proof than ever before.
The universally valid uncertainty relations derived by Branciard \cite{12EDU,13VHE} 
for pure states and by one of the authors \cite{14EDR} for mixed states 
considerably strengthened Ozawa's original relation \cite{03UVR,03UPQ} 
and is considered as the strongest relation ever. 
Their generalizations to general $\sigma$-finite von Neumann algebras of course match 
our aim.
To achieve this purpose the theory of the standard form of von Neumann algebras plays 
a crucial role.
The bound, denoted by $D_{AB}$ below, in uncertainty relations
is written in the language of the theory.
We would like to emphasize that  
the theory of the standard form of von Neumann algebras can be
a more powerful tool than ever for quantitative analysis in quantum information theory. 
This study may be counted as an instance supporting the opinion that ``the fields of operator 
algebras and quantum information can benefit each other in multiple ways"  \cite[ll.~6--7]{KGR16}.

We adopt the (von Neumann) algebraic formulation of quantum theory herein.
Observables of a physical system are described
by self-adjoint operators affiliated with a von Neumann algebra,
and physical situations and experimental settings of the system
are described by normal states on the von Neumann algebra.
Let $\mathcal{M}$ be a von Neumann algebra on a Hilbert space $\mathcal{H}$
and $(S,\mathcal{F})$ a measurable space.
We assume that von Neumann algebras are $\sigma$-finite in the present paper.
Let $\mathcal{M}_{s.a.}$ denote the set of self-adjoint elements of $\mathcal{M}$.
Let $\mathcal{M}_\ast$ denote the predual of $\mathcal{M}$,
let $\mathcal{M}_{\ast,+}$ denote the set of positive elements of $\mathcal{M}_\ast$,
and $\mathcal{S}_n(\mathcal{M})$ the set of normal states on $\mathcal{M}$.
For every $\rho\in\mathcal{M}_\ast$ and $M\in\mathcal{M}$, let
$\langle \rho,M\rangle$ denote the pairing of $\mathcal{M}_\ast$ and $\mathcal{M}$, i.e.,
$\langle \rho,M\rangle=\rho(M)$.

In quantum mechanics, the observables are described by the von Neumann algebra 
$\mathcal{M}=\textrm{\boldmath $B$}(\mathcal{H})$,
the set of bounded linear operators on a Hilbert space $\mathcal{H}$, 
and the set of states, described by the density operators on $\mathcal{H}$, corresponds 
to the set $\mathcal{S}_n(\mathcal{M})$ of normal states on $\textrm{\boldmath $B$}(\mathcal{H})$.
In the modern quantum mechanics having abandoned 
the repeatability hypothesis,
the Heisenberg type uncertainty relation
\begin{equation}
\varepsilon(A)\ep(B)\geq C_{AB}
\end{equation}
for an arbitrary pair of observables $A,B$ is known to hold only for a limited class of measurements,
for instance, for jointly unbiased joint measurements \cite{AK65,AG88,Ish91,91QU}, 
where $\varepsilon$ denotes the measurement error (to be defined in Section 3)
and $C_{AB}=C_{A,B,\rho}$ is defined by
\begin{equation}
 C_{A,B,\rho}=\frac{1}{2}|\langle\rho,-i[A,B]\rangle|.
\end{equation}
In 2003, one of the authors \cite{03UVR,03UPQ} derived a universally valid uncertainty relation
\begin{equation}\label{Ozawa}
\varepsilon(A)\ep(B)+\varepsilon(A)\sigma(B)+\sigma(A)\ep(B)\geq C_{AB},
\end{equation}
where $\sigma(A)=\sigma(A;\rho)$ is
the standard deviation of an observable $A$ in a normal state $\rho$ defined by
\begin{equation}
\sigma(A;\rho)=(\rho(A^2)-\rho(A)^2)^{\frac{1}{2}}=\langle\rho,(A-\rho(A))^2\rangle^{\frac{1}{2}}.
\end{equation}
In 2013, Branciard \cite{Bra13,Bra14} strengthened the above relation as
\begin{equation}\label{eq:Branc}
\varepsilon(A)^2\sigma(B)^2+\sigma(A)^2\ep(B)^2+2\varepsilon(A)\ep(B)
\sqrt{\sigma(A)^2\sigma(B)^2-C_{AB}^2}
\geq C_{AB}^2.
\end{equation}
Subsequently, this relation \eq{Branc} was further improved by one of the authors
\cite{14EDR} replacing the lower bound
$C_{AB}$ by a more stringent one $D_{AB}$ as 
\begin{equation}\label{Ozawa2}
\varepsilon(A)^2\sigma(B)^2+\sigma(A)^2\ep(B)^2+2\varepsilon(A)\ep(B)
\sqrt{\sigma(A)^2\sigma(B)^2-D_{AB}^2}
\geq D_{AB}^2,
\end{equation}
where $D_{AB}=D_{A,B,\rho}$ is defined by
\begin{equation}
D_{A,B,\rho}= \frac{1}{2}\mathrm{Tr}[|\sqrt{\tilde{\rho}}(-i[A,B])\sqrt{\tilde{\rho}}|],
\end{equation}
and $\tilde{\rho}$ is a (unique) density operator on $\mathcal{H}$ such that $\rho(M)=\mathrm{Tr}[M \tilde{\rho}]$
for all $M\in\textrm{\boldmath $B$}(\mathcal{H})$.

The fact that the bound $D_{AB}$ 
depends on the choice of the observable algebra relevant to the measuring
interaction is important.
This is easily seen by the example below.
\begin{example}[]
Let $\mathcal{N}=\mathrm{M}_2(\mathbb{C})\otimes\mathrm{M}_2(\mathbb{C})\cong\mathrm{M}_4(\mathbb{C})$
be a von Neumann algebra as the observable algebra,
 $\sigma_x$, $\sigma_y$ and $\sigma_z$ the Pauli matrices
and $\omega_\psi=\langle \psi| (\cdot)\psi \rangle$ a normal state on
$\mathrm{M}_2(\mathbb{C})\otimes\mathrm{M}_2(\mathbb{C})\cong\mathrm{M}_4(\mathbb{C})$,
where $\psi$ is defined by
\begin{equation}
\psi=\sqrt{\dfrac{1}{2}}(e_{z\uparrow}\otimes e_{z\downarrow}- e_{z\downarrow}\otimes e_{z\uparrow}),
\end{equation}
and $e_{z\uparrow}$ and $e_{z\downarrow}$ are eigenvectors of $\sigma_z$
corresponding to eigenvalues $+1$ and $-1$, respectively,
i.e., $\sigma_z e_{z\uparrow}= e_{z\uparrow}$, $\sigma_z e_{z\downarrow}=- e_{z\downarrow}$.
Here, we set
$\mathcal{M}=\mathrm{M}_2(\mathbb{C})\otimes\mathbb{C}1\cong\mathrm{M}_2(\mathbb{C})$,
which is of course a von Neumann subalgebra of
$\mathcal{N}=\mathrm{M}_2(\mathbb{C})\otimes\mathrm{M}_2(\mathbb{C})\cong\mathrm{M}_4(\mathbb{C})$.
Then the restriction $\omega_\psi|_{\mathrm{M}_2(\mathbb{C})}$ of $\omega_\psi$
to $\mathrm{M}_2(\mathbb{C})$ satisfies
\begin{equation}
\omega_\psi|_{\mathrm{M}_2(\mathbb{C})}(M)=
\langle \psi|(M\otimes 1)\psi \rangle=\dfrac{1}{2}\mathrm{Tr}[M]
\end{equation}
for all $M\in \mathrm{M}_2(\mathbb{C})$. 
This is an example that every normal state on a von Neumann algebra is defined by a vector state
on a larger von Neumann algebra than the original one.
Then we have
\begin{align}
D_{\sigma_x\otimes 1,\sigma_y\otimes 1,\omega_\psi}&=
C_{\sigma_x\otimes 1,\sigma_y\otimes 1,\omega_\psi}=
\dfrac{1}{2}|\langle \psi|[\sigma_x\otimes 1,\sigma_y\otimes 1]\psi \rangle|
=\dfrac{1}{2}|\langle \psi|2i(\sigma_z\otimes 1)\psi \rangle|=0, \\
D_{\sigma_x,\sigma_y,\frac{1}{2}\mathrm{Tr}}&=\dfrac{1}{2}\mathrm{Tr}
\left| \sqrt{\dfrac{1}{2}}[\sigma_x,\sigma_y]\sqrt{\dfrac{1}{2}}\right|
=\dfrac{1}{2}\mathrm{Tr}
\left| \dfrac{1}{2} 2i\sigma_z\right|=\dfrac{1}{2}\mathrm{Tr}[1]=1.
\end{align}
Therefore, $D_{\sigma_x\otimes 1,\sigma_y\otimes 1,\omega_\psi}< D_{\sigma_x,\sigma_y,\frac{1}{2}\mathrm{Tr}}$.
That is, when we consider the system as a range where the observable algebra is $\mathcal{N}$, 
the bound value is $0$; however, when only observables contained in $\mathcal{M}$
are considered, the bound value is $1$.
\end{example}
Generally, it holds that the bound becomes smaller
as we consider systems with higher degrees of freedom.
This is because the class of physically admissible measurements is the wider
the more observables can be involved in the interaction Hamiltonian for the 
measurement.
Therefore, the designation of the observable algebra (of the system to be measured)
substantially contributes to the determination of the bound.
We would like to emphasize that the use of $D_{AB}$ is firmly valuable.
For more detailed discussions we refer the reader to \cite{14EDR,16A3}.

\sloppy
In this paper, we extend Eq.(\ref{Ozawa2}) to the setting of $\sigma$-finite 
von Neumann algebras.
Let $(\mathcal{M},\mathcal{H},\mathcal{P},J)$ be a standard form
with a $\sigma$-finite von Neumann algebra $\mathcal{M}$, i.e.,
$(\mathcal{M},\mathcal{H},\mathcal{P},J)$ is a quadruple consisting of
a Hilbert space $\mathcal{H}$,
a $\sigma$-finite von Neumann algebra $\mathcal{M}$ on $\mathcal{H}$,
a self-dual cone $\mathcal{P}$ of $\mathcal{H}$ 
and the modular conjugation $J$ of $\mathcal{M}$.
We refer the reader to Section \ref{2} (and \cite{BR02,Haa75,Tak02})
for the theory of the standard form of von Neumann algebras.
Let $A$ and $B$ be elements of $\mathcal{M}_{s.a.}$ 
and $\rho$ a normal state on $\mathcal{M}$.
We define a normal functional $\omega_{A,B,\rho}$ on $\mathcal{M}$ by
\begin{equation}
\omega_{A,B,\rho}(M)=\langle \xi_\rho|M J(-i[A,B])J\xi_\rho\rangle
\end{equation}
for all $M\in\mathcal{M}$, where $\xi_\rho$ is a unique unit vector of $\mathcal{P}$ such that
$\langle \rho,M\rangle=\langle \xi_\rho|M \xi_\rho\rangle$
for all $M\in\mathcal{M}$. Here we redefine $D_{A,B,\rho}$ by
\begin{equation}
 D_{A,B,\rho}= \frac{1}{2}\Vert \omega_{A,B,\rho}\Vert.
\end{equation}
This coincides with the original one in the case of $\mathcal{M}=\textrm{\boldmath $B$}(\mathcal{H})$.
It is then obvious that
\begin{equation}
D_{A,B,\rho}=\frac{1}{2}\Vert \omega_{A,B,\rho}\Vert\geq \frac{1}{2}|\langle \omega_{A,B,\rho},1 \rangle|
=\frac{1}{2}|\langle \xi_\rho|J(-i[A,B])J\xi_\rho\rangle|=C_{A,B,\rho}.
\end{equation}
Our purpose is to derive Eq.(\ref{Ozawa2}) in terms of $D_{A,B,\rho}$ redefined above
in the setting  of general $\sigma$-finite von Neumann algebras.
For the original case of $\mathcal{M}=\textrm{\boldmath $B$}(\mathcal{H})$,
Ozawa \cite{14EDR} presented two proofs of the derivation of Eq.(\ref{Ozawa2}):
One is based on the method, called the ``canonical purification", in terms of
the dual Hilbert space $\mathcal{H}^\ast$ of $\mathcal{H}$,
and the other is based on the representation of $\textrm{\boldmath $B$}(\mathcal{H})$
on the Hilbert space of Hilbert-Schmidt class operators on $\mathcal{H}$.
Our proof herein is a natural unification of those methods
via the theory of the standard form of von Neumann algebras
and can also be applied to any measurements described
by CP instruments which cannot be realized by any measuring processes.

In Section \ref{2}, we introduce quantum measurement theory 
for quantum systems described by general $\sigma$-finite von Neumann algebras.
In Section \ref{3}, we define the error and the disturbance of measurements 
used in this paper.
In Section \ref{4}, we show the main theorem of the paper, that is,
a universally valid uncertainty relation for measurement error and disturbance
in general quantum systems described by $\sigma$-finite von Neumann algebras.
In Section \ref{5}, we also show a universally valid uncertainty relation 
for simultaneous measurements.
In Section \ref{6}, we conclude the paper with some remarks.

\section{Preliminaries on Quantum Measurement}\label{2}

Here we introduce quantum measurement theory based on completely positive (CP) instruments defined on
general ($\sigma$-finite) von Neumann algebras,
which enables us to describe processes of measurement in
quantum systems with infinite degrees of freedom, expecially, in quantum fields.
The previous investigation \cite{16A1} by the authors
much contributes to the development of the theory and it mathematics.
Our attempt herein, the establishment of universally valid uncertainty relations
in general quantum systems, is its succeeding program.
Extending the scope of application of uncertain relations to quantum fields
is essential for developing both foundations of quantum theory
and quantitative analysis in quantum field theory.
In particular, many physicists, inspired by quantum information theory,
are recently very interested in the latter.
Therefore, our study has potential demand in physics and is not just mathematical concern.
This section provides us with preliminaries on recent quantum measurement theory enough to understand
the physical setting and mathematical proof of universally valid uncertainty relations.
We refer the reader to \cite{04URN,14MFQ,16A1} 
for detailed expositions of quantum measurement theory
based on CP instruments and measuring processes.

First we shall define the concept of CP instrument
describing output probabilities and dynamical changes of states caused by physically realizable measurement.
Let $\mathcal{M}$ be a von Neumann algebra on a Hilbert space $\mathcal{H}$
and $(S,\mathcal{F})$ a measurable space.
Let $\mathrm{P}(\mathcal{M}_\ast)$ denote the set of positive linear maps on $\mathcal{M}_\ast$.
\begin{definition}[Instrument \text{\cite{DL70,Dav76,84QC}}]
A map $\mathcal{I}:\mathcal{F}\rightarrow \mathrm{P}(\mathcal{M}_\ast)$ is called
an instrument $\mathcal{I}$ for $(\mathcal{M},S)$
if it satisfies the following two conditions:\\
$(1)$ $\Vert\mathcal{I}(S)\rho\Vert=\Vert\rho\Vert$ for all $\rho\in \mathcal{M}_{\ast,+}$;\\
$(2)$ For every $\rho\in\mathcal{M}_\ast$, $M\in\mathcal{M}$ and
countable mutually disjoint sequence $\{\Delta_j\}\subset\mathcal{F}$,
\begin{equation}
\langle \mathcal{I}(\cup_j \Delta_j)\rho, M \rangle
=\sum_j \langle \mathcal{I}(\Delta_j)\rho, M \rangle.
\end{equation}
An instrument $\mathcal{I}$ for $(\mathcal{M},S)$ is called a completely positive instrument,
or a CP instrument for short, if $\mathcal{I}(\Delta)$ is completely positive for every $\Delta\in\mathcal{F}$.
\end{definition}
An instrument $\mathcal{I}$ for $(\mathcal{M},S)$ represents a measuring apparatus 
${\bf A}({\bf x})$ with output variable ${\bf x}$ taking values in the measurable space $(S,\mathcal{F})$,
and specifies both the probability measure $\mathrm{Pr}\{{\bf x}\in(\cdot)\Vert\rho\}$
of ${\bf x}$ and the family $\{\rho_{\{{\bf x}\in\Delta\}}\}_{\Delta\in\mathcal{F}}$
of states after the measurement in each normal state $\rho$ on $\mathcal{M}$, which are given by
\begin{align}
\mathrm{Pr}\{{\bf x}\in\Delta\Vert\rho\}&=\Vert\mathcal{I}(\Delta)\rho\Vert,\hspace{5mm}\Delta\in\mathcal{F},\\
\rho_{\{{\bf x}\in\Delta\}}&=
\left\{
\begin{array}{ll}
\dfrac{\mathcal{I}(\Delta)\rho}{\Vert\mathcal{I}(\Delta)\rho\Vert},
&\quad(\mathrm{if}\;\mathrm{Pr}\{{\bf x}\in\Delta\Vert\rho\}>0),\\
\mathcal{I}(S)\rho, &\quad(\mathrm{otherwise}),
\end{array}
\right.
\end{align}
respectively, where each state $\rho_{\{{\bf x}\in\Delta\}}$ realizes after the measurement
when $\rho$ is prepared before the measurement and output values of ${\bf x}$
not contained in $\Delta$ is ignored during data processing.
Conversely, an instrument $\mathcal{I}$ for $(\mathcal{M},S)$ is defined by
a measuring apparatus ${\bf A}({\bf x})$ with
output variable ${\bf x}$ taking values in the measurable space $(S,\mathcal{F})$
if for all $M\in\mathcal{M}_{s.a.}$
the joint probability measure $\mathrm{Pr}\{(M,{\bf x})\in\Delta\Vert\rho\}$ 
on $(\mathbb{R}\times S,\mathcal{B}(\mathbb{R})\times\mathcal{F})$
of the successive measurement carried out by ${\bf A}({\bf x})$ and
the measurement of $M$ in this order
is an affine function of $\rho\in\mathcal{S}_n(\mathcal{M})$.
Here, the map $\Gamma\times\Delta\in\mathcal{B}(\mathbb{R})\times\mathcal{F}
 \mapsto\mathrm{Pr}\{(M,{\bf x})\in\Gamma\times\Delta\Vert\rho\}$ is defined by
\begin{equation}
\mathrm{Pr}\{M\in\Gamma,{\bf x}\in\Delta\Vert\rho\}=\mathrm{Pr}\{{\bf x}\in\Delta\Vert\rho\}
\mathrm{Pr}\{M\in\Gamma\Vert\rho_{\{{\bf x}\in\Delta\}}\}
\end{equation}
and is then uniquely extended into the probability measure on 
$(\mathbb{R}\times S,\mathcal{B}(\mathbb{R})\times\mathcal{F})$, where 
$\mathrm{Pr}\{M\in\Gamma,{\bf x}\in\Delta\Vert\rho\}=\mathrm{Pr}\{(M,{\bf x})\in\Gamma\times\Delta\Vert\rho\}$.
In addition, it is known that an instrument is CP if and only if
the measuring apparatus that defines the given instrument satisfies
 the condition called the trivial extendability \cite{04URN,14MFQ,16A1}.

Now we consider a map $\mathcal{I}:\mathcal{M}\times\mathcal{F}\rightarrow \mathcal{M}$
satisfying the following three conditions:\\
$(i)$ For every $\Delta\in\mathcal{F}$, $M\mapsto \mathcal{I}(\Delta,M)$ is
a normal positive linear map on $\mathcal{M}$.\\
$(ii)$ $\mathcal{I}(1,S)=1$.\\
$(iii)$ For every $\rho\in\mathcal{M}_\ast$, $M\in\mathcal{M}$ and
countable mutually disjoint sequence $\{\Delta_j\}\subset\mathcal{F}$,
\begin{equation}
\langle \rho, \mathcal{I}(M,\cup_j \Delta_j) \rangle
=\sum_j\langle \rho, \mathcal{I}(M,\Delta_j)\rangle.
\end{equation}
There is a one-to-one correspondence between an instrument for $(\mathcal{M},S)$
and a map $\mathcal{I}:\mathcal{M}\times\mathcal{F}\rightarrow \mathcal{M}$
satisfying the above three conditions, which is given by the relation
\begin{equation}
\langle \mathcal{I}(\Delta)\rho, M \rangle=\langle \rho, \mathcal{I}(M,\Delta)\rangle
\end{equation}
for all $\rho\in\mathcal{M}_\ast$, $M\in\mathcal{M}$ and $\Delta\in\mathcal{F}$.
Thus we also call the map $\mathcal{I}:\mathcal{M}\times\mathcal{F}\rightarrow \mathcal{M}$
an instrument $(\mathcal{M},S)$.
We then define a probability operator-valued measure $\Pi_\mathcal{I}$ 
associated with an instrument $\mathcal{I}$ for $(\mathcal{M},S)$ by
$\Pi_\mathcal{I}(\Delta)=\mathcal{I}(1,\Delta)$ for all $\Delta\in\mathcal{F}$.

Let $\Pi:\mathcal{B}(\mathbb{R})\rightarrow\mathcal{M}$ be a probability operator-valued measure.
For every $n\in\mathbb{N}$, we define the symmetric operator $\Pi^{(n)}$ by
\begin{equation}
\langle\xi|\Pi^{(n)}\eta\rangle = \int x^n\;d\langle \xi|\Pi(x)\eta\rangle
\end{equation}
for any $\xi,\eta\in\mathrm{dom}(\Pi^{(n)})$, where we define the domain $\mathrm{dom}(\Pi^{(n)})$ by
\begin{equation}
\mathrm{dom}(\Pi^{(n)})=\{\xi\in\mathcal{H}\;|\;\int_\mathbb{R}x^{2n}\;d\langle \xi| \Pi(x)\xi\rangle
<\infty\}.
\end{equation}

Every CP instrument admits the following representation theorem.
\begin{proposition}[\text{\cite[Proposition 4.2]{84QC}}]\label{Naimark-Ozawa}
For any CP instrument $\mathcal{I}$ for $(\mathcal{M},S)$, 
there are a Hilbert space $\mathcal{K}$,
a spectral measure $E:\mathcal{F}\rightarrow \textrm{\boldmath $B$}(\mathcal{K})$,
a nondegenerate normal representation
$\pi:\mathcal{M}\rightarrow\textrm{\boldmath $B$}(\mathcal{K})$
and an isometry $V\in\textrm{\boldmath $B$}(\mathcal{H},\mathcal{K})$ satisfying
\begin{align}
\mathcal{I}(M,\Delta) &=V^\ast \pi(M)E(\Delta)V, \label{CPrep}\\
E(\Delta)\pi(M) &=\pi(M)E(\Delta)
\end{align}
for any $\Delta\in\mathcal{F}$ and $M\in\mathcal{M}$,
and $\mathcal{K}=\overline{\mathrm{span}}(\pi(\mathcal{M})E(\mathcal{F})V\mathcal{H}))$.
\end{proposition}
The quadruple $(\mathcal{K},E,\pi,V)$ in the above proposition is unique up to unitary equivalence.
We call the quadruple $(\mathcal{K},E,\pi,V)$ a minimal dilation of $\mathcal{I}$.

Let $\mathcal{X}$ and $\mathcal{Y}$ be C$^\ast$-algebras on Hilbert spaces $\mathcal{H}$ and $\mathcal{K}$,
respectively.
The minimal tensor product of $\mathcal{X}$ and $\mathcal{Y}$, denoted by
$\mathcal{X}\otimes_{\mathrm{min}}\mathcal{Y}$, is defined by the completion of
the algebraic tensor product $\mathcal{X}\otimes_{\mathrm{alg}}\mathcal{Y}$
of $\mathcal{X}$ and $\mathcal{Y}$
(as a $^\ast$-subalgebra of $\textrm{\boldmath $B$}(\mathcal{H}\otimes\mathcal{K})$)
 by the norm topology of $\textrm{\boldmath $B$}(\mathcal{H}\otimes\mathcal{K})$.
Let $\mathcal{M}$ and $\mathcal{N}$ be von Neumann algebras on Hilbert spaces $\mathcal{H}$ and $\mathcal{K}$,
respectively. The $W^\ast$-tensor product of $\mathcal{M}$ and $\mathcal{M}$, denoted by
$\mathcal{M}\;\overline{\otimes}\;\mathcal{N}$, is defined by the completion of
$\mathcal{M}\otimes_{\mathrm{alg}}\mathcal{N}$ by the ultraweak topology 
of $\textrm{\boldmath $B$}(\mathcal{H}\otimes\mathcal{K})$.
We refer the reader to \cite{EL77,Tak79} for details on tensor products of operator algebras.

By Proposition \ref{Naimark-Ozawa} \cite[Proposition 3.3]{16A1},
for every CP instrument $\mathcal{I}$ for $(\mathcal{M},S)$,
there exists a unique unital CP map
$\Psi_\mathcal{I}:\mathcal{M}\otimes_{\mathrm{min}} L^\infty(S,\mathcal{I})\rightarrow\mathcal{M}$
such that $\Psi_\mathcal{I}(M\otimes[\chi_\Delta])=\mathcal{I}(M,\Delta)$
for all $M\in\mathcal{M}$ and $\Delta\in\mathcal{F}$.

Next, we shall define the concept of measuring process,
which is nothing but a quantum mechanical modeling of measuring apparatus.
Let $\mathcal{M}$ and $\mathcal{N}$ be von Neumann algebras.
For every $\sigma\in \mathcal{N}_\ast$, the map
$\mathrm{id}\otimes\sigma:\mathcal{M}\;\overline{\otimes}\;\mathcal{N}
\rightarrow\mathcal{M}$ is defined by
$\langle \rho\otimes\sigma,X \rangle=\langle \rho,(\mathrm{id}\otimes\sigma)(X) \rangle$
for all $X\in\mathcal{M}\;\overline{\otimes}\;\mathcal{N}$
and $\rho\in\mathcal{M}_\ast$.

\begin{definition}[Measuring process \text{\cite[Definition 3.4]{16A1}}]
A measuring process $\mathbb{M}$ for $(\mathcal{M},S)$
is a quadruple $\mathbb{M}=(\mathcal{K},\sigma,F,U)$
consisting of a Hilbert space $\mathcal{K}$, a normal
state $\sigma$ on $\textrm{\boldmath $B$}(\mathcal{K})$,
a spectral measure $F:\mathcal{F}\rightarrow \textrm{\boldmath $B$}(\mathcal{K})$,
and a unitary operator $U$ on $\mathcal{H}\otimes\mathcal{K}$
satisfying 
\begin{equation}
\{\mathcal{I}_\mathbb{M}(M,\Delta)\;|\;M\in\mathcal{M},\Delta\in\mathcal{F}\}\subset\mathcal{M},
\end{equation}
 where $\mathcal{I}_\mathbb{M}$ is a CP instrument for $(\textrm{\boldmath $B$}(\mathcal{H}),S)$
defined by $\mathcal{I}_\mathbb{M}(X,\Delta)=(\mathrm{id}\otimes\sigma)[U^\ast(X\otimes F(\Delta))U]$
for all $X\in\textrm{\boldmath $B$}(\mathcal{H})$ and $\Delta\in\mathcal{F}$.
\end{definition}

\begin{definition}[Statistical equivalence class of measuring processes \cite{84QC}]
Two measuring processes $\mathbb{M}_1=(\mathcal{K}_1,\sigma_1,F_1,U_1)$
and $\mathbb{M}_2=(\mathcal{K}_2,\sigma_2,F_2,U_2)$ for $(\mathcal{M},S)$
are said to be statistically equivalent if
$\mathcal{I}_{\mathbb{M}_1}(M,\Delta)=\mathcal{I}_{\mathbb{M}_2}(M,\Delta)$
 for all $M\in\mathcal{M}$ and $\Delta\in\mathcal{F}$.
\end{definition}

MO established the following one-to-one correspondence for the case of
$\mathcal{M}=\textrm{\boldmath $B$}(\mathcal{H})$.

\begin{theorem}[\text{\cite[Theorem 5.1]{84QC}}]
Let $\mathcal{H}$ be a Hilbert space and $(S,\mathcal{F})$ a measurable space.
Then there is a one-to-one correspondence
between statistical equivalence classes of measuring processes $\mathbb{M}=(\mathcal{K},\sigma,F,U)$
for $(\textrm{\boldmath $B$}(\mathcal{H}),S)$
and CP instruments $\mathcal{I}$ for $(\textrm{\boldmath $B$}(\mathcal{H}),S)$,
which is given by the relation $\mathcal{I}(X,\Delta)=\mathcal{I}_\mathbb{M}(X,\Delta)$
for all $\Delta\in\mathcal{F}$ and $X\in\textrm{\boldmath $B$}(\mathcal{H})$.
\end{theorem}
This theorem states that every CP instrument is modeled by a measuring process and
every measuring process defines a CP instrument.
To generalize the above theorem to general $\sigma$-finite von Neumann algebras,
we define the following property for CP instruments.

\begin{definition}[Normal extension property \text{\cite[Definition 3.4]{16A1}}]
Let $\mathcal{I}$ be a CP instrument for $(\mathcal{M},S)$.
$\mathcal{I}$ has the normal extension property (NEP) if there exists a unital normal CP map
$\widetilde{\Psi_\mathcal{I}}:\mathcal{M}\;\overline{\otimes}\; L^\infty(S,\mathcal{I})\rightarrow
\textrm{\boldmath $B$}(\mathcal{H})$ such that
$\widetilde{\Psi_\mathcal{I}}|_{\mathcal{M}\otimes_{\mathrm{min}}L^\infty(S,\mathcal{I})}=\Psi_\mathcal{I}$.
\end{definition}
The authors established the following theorem in \cite{16A1}.

\begin{theorem}[\text{\cite[Theorem 3.2]{16A1}}]\label{mainse4}
Let $\mathcal{M}$ be a von Neumann algebra on a Hilbert space
$\mathcal{H}$ and $(S,\mathcal{F})$ a measurable space.
Then there is a one-to-one correspondence
between statistical equivalence classes of measuring processes $\mathbb{M}=(\mathcal{K},\sigma,F,U)$
for $(\mathcal{M},S)$ and CP instruments $\mathcal{I}$ for $(\mathcal{M},S)$ \textbf{with the NEP},
which is given by the relation $\mathcal{I}(M,\Delta)=\mathcal{I}_{\mathbb{M}}(M,\Delta)$
for all $\Delta\in\mathcal{F}$ and $M\in\mathcal{M}$.
\end{theorem}
Thus every CP instrument defined on general $\sigma$-finite von Neumann algebras
is not always realized by a measuring process.
The followings are examples of CP instruments without the NEP given in \cite[Section 4]{16A1}.
\begin{example}[\textrm{\cite[Example 5.1]{16A1}}]\label{NoNEP1}
Let $m$ be Lebesgue measure on $[0,1]$. A CP instrument $\mathcal{I}_m$ for
$(L^{\infty}([0,1],m), [0,1])$ is defined by $\mathcal{I}_m(f,\Delta)=[\chi_\Delta] f$
for all $\Delta\in\mathcal{B}([0,1])$ and $f\in L^{\infty}([0,1],m)$.
\end{example}
\begin{example}[\textrm{\cite[Example 5.2]{16A1}}]\label{NoNEP2}
Let $\mathcal{R}$ be an approximately finite-dimensional (AFD) von Neumann algebra of type $\mathrm{II}_1$
on a separable Hilbert space $\mathcal{H}$.
Let $A$ be a self-adjoint operator with continuous spectrum affiliated with $\mathcal{R}$
and $\mathcal{E}$ a (normal) conditional expectation of $\mathcal{R}$ onto
$\{A\}^\prime\cap\mathcal{R}$,
where $\{A\}^\prime=\{E^{A}(\Delta)\;|\; \Delta\in \mathcal{B}(\mathbb{R})\}^\prime$.
A CP instrument $\mathcal{I}_A$ for $(\mathcal{R},\mathbb{R})$ is defined by 
\begin{equation}\label{repeatable}
\mathcal{I}_A(M,\Delta)=\mathcal{E}(M)E^A(\Delta)
\end{equation}
for all $M\in\mathcal{R}$ and $\Delta\in\mathcal{B}(\mathbb{R})$.
\end{example}
It also is shown in \cite[Section 4]{16A1}
that every CP instrument defined on atomic von Neumann algebras has the NEP.
By contrast, CP instruments without the NEP are defined on non-atomic (injective) von Neumann algebras.
Let $\mathcal{M}$ be an injective von Neumann algebra on a Hilbert space $\mathcal{H}$
and $(S,\mathcal{F})$ a measurable space.
For every CP instrument $\mathcal{I}$ for $(\mathcal{M},S)$
there exists a net $\{\mathcal{I}_\alpha\}_{\alpha\in A}$
of CP instruments for $(\mathcal{M},S)$ with the NEP such that $\mathcal{I}_\alpha$
ultraweakly converges to $\mathcal{I}$ and that $\mathcal{I}(1,\Delta)=\mathcal{I}_\alpha(1,\Delta)$
for all $\alpha$ and $\Delta\in\mathcal{F}$ \cite[Section 4]{16A1}.
Since local algebras in algebraic quantum field theory are injective,
we can apply these results to the characterization of local measurements of quantum fields
\cite[Section 6]{16A1}.

Lastly, we shall introduce the theory of standard forms of 
von Neumann algebras.
Let $\mathcal{H}$ be a Hilbert space. 
For every subset $S$ of $\textrm{\boldmath $B$}(\mathcal{H})$,
let $S^\prime$ denote the commutant of $S$, i.e., 
$S^\prime=\{A\in\textrm{\boldmath $B$}(\mathcal{H})\;|\;AB=BA\;
\mathrm{for}\;\mathrm{all}\;B\in S\}$.
We call a convex subset $\mathcal{P}$ of $\mathcal{H}$
a cone of $\mathcal{H}$. For every cone $\mathcal{P}$ of $\mathcal{H}$, we define
the dual cone $\mathcal{P}^\vee$ of $\mathcal{P}$ by
\begin{equation}
\mathcal{P}^\vee=\{\xi\in\mathcal{H}\;|\;\langle \xi| \eta\rangle\geq 0\;
\mathrm{for}\;\mathrm{all}\;\eta\in\mathcal{P}\}.
\end{equation}
A cone $\mathcal{P}$ of $\mathcal{H}$ is said to be self-dual if $\mathcal{P}^\vee=\mathcal{P}$.
For a von Neumann algebra $\mathcal{M}$ on a Hilbert space $\mathcal{H}$,
a quadruple $(\mathcal{M},\mathcal{H},\mathcal{P},J)$ of an anti-linear isometry $J$ with $J^2=1$,
called the modular conjugation $J$ of $\mathcal{M}$,
and a self-dual cone $\mathcal{P}$ of $\mathcal{H}$ is called a standard form of $\mathcal{M}$
if it satisfies the following four conditions:\\
$(1)$ $J\mathcal{M}J=\mathcal{M}^\prime$;\\
$(2)$ $J\xi=\xi$ for all $\xi\in\mathcal{P}$;\\
$(3)$ $MJMJ\mathcal{P}\subset\mathcal{P}$ for all $M\in\mathcal{M}$;\\
$(4)$ $JZJ=Z^\ast$ for all $Z\in\mathcal{Z}(\mathcal{M})=\mathcal{M}\cap\mathcal{M}^\prime$.\\
In fact, it is shown in \cite[Lemma 3.19]{AH14} that the fourth condition is redundant.
Namely, a quadruple $(\mathcal{M},\mathcal{H},\mathcal{P},J)$
satisfying the conditions $(1)$, $(2)$ and $(3)$ is a standard form of $\mathcal{M}$.
For every von Neumann algebra $\mathcal{N}$ on a Hilbert space $\mathcal{K}$,
there exists a standard form $(\mathcal{M},\mathcal{H},\mathcal{P},J)$ such that
$\mathcal{N}$ is $W^\ast$-isomorphic to $\mathcal{M}$.
Therefore, we assume that von Neumann algebras appearing in the paper
are in standard forms without loss of generality.

\section{Error and Disturbance}\label{3}

In this section, we define an error and a disturbance in quantum measurement theory,
which are introduced and applied in past investigations \cite{88MS,89RS,03UVR,03UPQ,04URN}.
The concept of error is fundamental for measuring how accurately a measuring apparatus can measure an observable,
and is defined (or characterized) as the root-mean-square of
the ``difference" between an observable to be measured
and an output variable of the measuring apparatus as an observable actually measured.
On the other hand, the concept of disturbance is essential for
estimating the effect of measurement on the system to be measured,
and is defined as the root-mean-square of the ``difference" between an observable 
before the measurement and the identical one after the measurement.
It should be noted that the disturbance of an observable $B$ caused by a measurement 
of $A$ can be defined as the error of such a measurement of $B$ in the state just before
the $A$-measurement that is actually carried out by the precise $B$-measuring 
apparatus just after the $A$-measurement \cite{03UVR,03UPQ}.
These quantities have been studied from various perspectives,
therefore, we will mention them minimally herein.

Let $\mathcal{M}$ be a von Neumann algebra on a Hilbert space $\mathcal{H}$
and $A$, $B$ self-adjoint elements of $\mathcal{M}$.
Let $\mathcal{I}$ be a CP instrument for $(\mathcal{M},\mathbb{R})$.
$\mathcal{I}$ physically corresponds to a measuring apparatus with output variable ${\bf x}$ taking values in
$(\mathbb{R},\mathcal{B}(\mathbb{R}))$.
We then define an error $\varepsilon(A)=
\varepsilon(A,\rho;\mathcal{I})$ of measurement of $A$ in $\rho$
and a disturbance $\eta(B)=\eta(B,\rho;\mathcal{I})$ of $B$ in $\rho$ by
\begin{align}
\varepsilon(A,\rho;\mathcal{I}) &=
\langle \rho,\Pi_\mathcal{I}^{(2)}-A\Pi_\mathcal{I}^{(1)}-\Pi_\mathcal{I}^{(1)}A+A^2 \rangle^{\frac{1}{2}},\\
\eta(B,\rho;\mathcal{I}) &=
\langle \rho, \mathcal{I}(B^2,\mathbb{R})-B\mathcal{I}(B,\mathbb{R})-
\mathcal{I}(B,\mathbb{R})B+B^2\rangle^{\frac{1}{2}},
\end{align}
respectively.
If $\mathcal{I}$ is realized by a measuring process
$(\mathcal{K},\sigma,F,U)$ for $(\mathcal{M},\mathbb{R})$,
it holds that
\begin{align}
\varepsilon(A,\rho;\mathcal{I}) &=\langle 
\tilde{\rho}\otimes\sigma,(U^\ast(1\otimes F^{(1)})U-A\otimes1)^2 \rangle^{\frac{1}{2}},\\
\eta(B,\rho;\mathcal{I}) &=\langle 
\tilde{\rho}\otimes\sigma,(U^\ast(B\otimes 1)U-B\otimes1)^2 \rangle^{\frac{1}{2}},
\end{align}
where $\tilde{\rho}$ is a normal state on $\textrm{\boldmath $B$}(\mathcal{H})$
such that $\rho(M)=\tilde{\rho}(M)$ for all $M\in\mathcal{M}$.
The operator $U^\ast(1\otimes F^{(1)})U-A\otimes1$ is often called the noise
operator of the measuring process $(\mathcal{K},\sigma,F,U)$ in measuring $A$ and the error 
measure $\varepsilon$ is often called the noise-operator-based error.
Both $\varepsilon(A,\rho;\mathcal{I})$ and $\eta(B,\rho;\mathcal{I})$ 
are considered as 
a natural generalization of
those in classical probability theory.
It should also be remarked that the definitions of error and disturbance
are independent of the existence and the choice of measuring processes
which realize the given CP instrument.
A justification of the use of the noise-operator-based error was discussed
extensively in Ref.~\cite{18PPT}; it will be briefly discussed also in the last 
section.

Let $\mathcal{I}$ be a CP instrument for $(\mathcal{M},\mathbb{R})$ and
$(\mathcal{K},E,\pi,V)$ the minimal dilation of $\mathcal{I}$. 
We assume that $E^{(1)}$ is a bounded operator on $\mathcal{K}$,
so that $E^{(2)}=(E^{(1)})^2$.
Then we have
\begin{align}
\varepsilon(A,\rho;\mathcal{I}) &=
\langle \rho,(E^{(1)}V-VA)^\ast (E^{(1)}V-VA)\rangle^{\frac{1}{2}},\\
\eta(B,\rho;\mathcal{I}) &=
\langle \rho,(\pi(B)V-VB)^\ast (\pi(B)V-VB) \rangle^{\frac{1}{2}}.
\end{align}

\section{Main Theorem}\label{4}
We are now ready to state the main theorem.
Let $(\mathcal{M},\mathcal{H},\mathcal{P},J)$ be a standard form
with a $\sigma$-finite von Neumann algebra $\mathcal{M}$.

\begin{theorem}[]\label{EDUR}
Let $A,B$ be elements of $\mathcal{M}_{s.a.}$, $\rho\in\mathcal{S}_n(\mathcal{M})$,
$\mathcal{I}$ a CP instrument for $(\mathcal{M},\mathbb{R})$ and
$(\mathcal{K},E,\pi,V)$ the minimal dilation of $\mathcal{I}$.
Assume that $E^{(1)}$ is a bounded operator on $\mathcal{K}$. Then we have
\begin{equation}
\varepsilon(A)^2\sigma(B)^2+\sigma(A)^2\eta(B)^2+2\varepsilon(A)\eta(B)
\sqrt{\sigma(A)^2\sigma(B)^2-D_{AB}^2}
\geq D_{AB}^2.
\end{equation}
\end{theorem}

Branciard \cite{Bra13} proved that
the following relation, called Branciard's geometric inequality, holds.
\begin{proposition}{}\label{Bra}
Let $\mathcal{L}$ be a real vector space with real inner product $(\cdot,\cdot)$.
 For any vectors
$\mathbf{a},\mathbf{b},\mathbf{m},\mathbf{n}\in\mathcal{L}$ with $\mathbf{m}\;\bot\;\mathbf{n}$,
\begin{equation} \label{Branciard}
\Vert \mathbf{a}-\mathbf{m}\Vert^2 \Vert \mathbf{b}\Vert^2
+\Vert \mathbf{a}\Vert^2 \Vert \mathbf{b}-\mathbf{n}\Vert^2
+2\Vert \mathbf{a}-\mathbf{m}\Vert \Vert \mathbf{b}-\mathbf{n}\Vert
\sqrt{\Vert \mathbf{a}\Vert^2 \Vert \mathbf{b}\Vert^2-(\mathbf{a},\mathbf{b})^2}
\geq (\mathbf{a},\mathbf{b})^2.
\end{equation}
\end{proposition}
The proof is given in \cite[SI Text, section B]{Bra13}.
For every linear functional $\omega$ on $\mathcal{M}$, the adjoint functional $\omega^\ast$ of $\omega$
is defined by
\begin{equation}
\omega^\ast(M)=\overline{\omega(M^\ast)}
\end{equation}
for all $M\in\mathcal{M}$ \cite{Tak79}. We say that a linear functional $\omega$ on $\mathcal{M}$ is hermitian
if $\omega^\ast=\omega$. For every normal functional $\omega$ on $\mathcal{M}$,
there exist the smallest projections $E$ and $F$ such that
\begin{equation}
\langle \omega, M \rangle=\langle \omega, M E\rangle,
\hspace{5mm}\langle \omega, N \rangle=\langle \omega, FN \rangle
\end{equation}
for all $M,N\in\mathcal{M}$. The projections $E$ and $F$ are called, respectively,
the left and right support projections of $\omega$ and denoted by $S_l(\omega)$ and $S_r(\omega)$.
If $\omega$ is hermitian, then $S_r(\omega)=S_l(\omega)$, so that they are written as $S(\omega)$ \cite{Tak79}.
\begin{proposition}[\text{\cite[Chapter III, Theorem 4.2 and its proof]{Tak79}}]
Let $\mathcal{M}$ be a von Neumann algebra.
For every normal linear functional $\omega$ on $\mathcal{M}$,
there exist a partial isometry $V\in\mathcal{M}$ and
a positive normal linear functional $\varphi$ on $\mathcal{M}$ such that
\begin{equation}\label{polar}
\langle \omega,M \rangle=\langle \varphi,MV \rangle
\hspace{5mm}V^\ast V=S(\varphi).
\end{equation}
Furthermore, we have $S(\varphi)=S_r(\omega)$ and $VV^\ast=S_l(\omega)$.
If $\omega$ is hermitian, $V$ is self-adjoint.
\end{proposition}
The expression of $\omega$ in Eq.(\ref{polar}) is called the polar decomposition of $\omega$
and $\varphi$ is called the absolute value of $\omega$ and denoted by $|\omega|$.
We use the above propositions to prove Theorem \ref{EDUR}.
\begin{proof}[Proof of Theorem \ref{EDUR}]
The norm $\Vert \omega_{A,B,\rho}\Vert$ of $\omega_{A,B,\rho}$ satisfies the equality
\begin{equation}
\Vert \omega_{A,B,\rho}\Vert=\langle  |\omega_{A,B,\rho}|,1\rangle.
\end{equation}
Since $\omega_{A,B,\rho}$ is hermitian,
there exists a self-adjoint partial isometry $W\in\mathcal{M}$ such that
\begin{equation}
\langle\omega_{A,B,\rho},M\rangle=\langle|\omega_{A,B,\rho}|,MW\rangle
\end{equation}
for all $M\in\mathcal{M}$, and that $W^\ast W=S(|\omega_{A,B,\rho}|)$.
Thus we have
\begin{equation}
\Vert \omega_{A,B,\rho}\Vert=\langle  \omega_{A,B,\rho},W\rangle=\langle \xi_\rho|WJ(-i[A,B])J\xi_\rho \rangle.
\end{equation}

To use Eq. (\ref{Branciard}), we define
a real inner product $\langle \cdot|\cdot \rangle_\mathbb{R}$ on $\mathcal{K}$ by
\begin{equation}
\langle x|y \rangle_\mathbb{R} =\mathrm{Re}\langle x|y \rangle
\end{equation}
for all $x,y\in\mathcal{K}$,
and put
\begin{align}
a &= V(A-\rho(A))JW\xi_\rho,\\
b &= -iV(B-\rho(B))J\xi_\rho,\\
m &= (E^{(1)}-\rho(A))VJW\xi_\rho,\\
n &= -i(\pi(B)-\rho(B))VJ\xi_\rho.
\end{align}
Let $\Vert\cdot\Vert_\mathbb{R}$ denote the norm of $\mathcal{K}$ as a real Hilbert space
induced by $\langle \cdot|\cdot \rangle_\mathbb{R}$.
Since $WJV^\ast(E^{(1)}-\rho(A))(\pi(B)-\rho(B))VJ$ is self-adjoint,
$\langle \xi_\rho|WJV^\ast(E^{(1)}-\rho(A))(\pi(B)-\rho(B))VJ\xi_\rho \rangle$ is real.
Thus we have
\begin{align}
\langle m|n \rangle_\mathbb{R} &=\mathrm{Re}\langle (E^{(1)}-\rho(A))VJW\xi_\rho
|-i(\pi(B)-\rho(B))VJ\xi_\rho \rangle \nonumber\\
 &= \mathrm{Re}\; (-i) \langle \xi_\rho|WJV^\ast(E^{(1)}-\rho(A))(\pi(B)-\rho(B))VJ\xi_\rho \rangle=0.
\end{align}
Next, $\langle a|b \rangle_\mathbb{R}$ and $\Vert a\Vert_\mathbb{R}^2$, respectively,
satisfy the following relations:
\begin{align}
\langle a|b \rangle_\mathbb{R} &=\mathrm{Re}\langle V(A-\rho(A))JW\xi_\rho
|-iV(B-\rho(B))J\xi_\rho \rangle \nonumber\\
 &= \mathrm{Re}\langle (A-\rho(A))JW\xi_\rho|-i(B-\rho(B))J\xi_\rho \rangle\nonumber\\
 &= \frac{\langle (A-\rho(A))JW\xi_\rho|-i(B-\rho(B))J\xi_\rho \rangle
 +\langle -i(B-\rho(B))J\xi_\rho|(A-\rho(A))JW\xi_\rho \rangle}{2} \nonumber\\
 &= \frac{\langle J(-i)(A-\rho(A))(B-\rho(B))JW\xi_\rho|\xi_\rho \rangle
 +\langle Ji(B-\rho(B))(A-\rho(A))JW\xi_\rho|\xi_\rho \rangle}{2}\nonumber\\
 &=\frac{1}{2}\langle \xi_\rho|WJ(-i[A,B])J\xi_\rho \rangle=\frac{1}{2}\Vert\omega_{A,B,\rho}\Vert
 =D_{AB}.
\end{align}
\begin{align}
\Vert a\Vert_\mathbb{R}^2=\Vert a\Vert^2 &=\Vert V(A-\rho(A))JW\xi_\rho\Vert^2 \nonumber\\
 &= \langle \xi_\rho|J(A-\rho(A))^2JW^2\xi_\rho\rangle \nonumber\\
 &\leq \langle\xi_\rho|J(A-\rho(A))^2J\xi_\rho\rangle= \rho((A-\rho(A))^2)=\sigma(A)^2.
\end{align}
Since $\Pi_\mathcal{I}^{(1)}=V^\ast E^{(1)}V$ and $\Pi_\mathcal{I}^{(2)}=V^\ast E^{(2)}V$
are bounded by assumption, we have
\begin{align}
\Vert a-m\Vert_\mathbb{R}^2=\Vert a-m\Vert^2 &=\Vert (VA-E^{(1)}V)JW\xi_\rho\Vert^2 \nonumber\\
 &= \langle \xi_\rho|J(\Pi_\mathcal{I}^{(2)}-A\Pi_\mathcal{I}^{(1)}-\Pi_\mathcal{I}^{(1)}A+A^2)J
 W^2\xi_\rho\rangle \nonumber\\
 &\leq \langle \xi_\rho|
 J(\Pi_\mathcal{I}^{(2)}-A\Pi_\mathcal{I}^{(1)}-\Pi_\mathcal{I}^{(1)}A+A^2)J\xi_\rho\rangle\nonumber\\
 &=\rho(\Pi_\mathcal{I}^{(2)}-A\Pi_\mathcal{I}^{(1)}-\Pi_\mathcal{I}^{(1)}A+A^2)=\varepsilon(A)^2.
\end{align}
Lastly, we have
\begin{align}
\Vert b\Vert_\mathbb{R}^2&=\Vert b\Vert^2=\Vert -iV(B-\rho(B))J\xi_\rho\Vert^2=\sigma(B)^2,\\
\Vert b-n\Vert_\mathbb{R}^2&=\Vert b-n\Vert^2=\Vert -i(VB-\pi(B)V)J\xi_\rho\Vert^2=\eta(B)^2.
\end{align}
Therefore, Eq. (\ref{Branciard}) implies
\begin{align}
 &\; \hspace{4mm}\varepsilon(A)^2\sigma(B)^2+\sigma(A)^2\eta(B)^2+2\varepsilon(A)\eta(B)
\sqrt{\sigma(A)^2\sigma(B)^2-D_{AB}^2}\nonumber\\
 &\geq  \Vert (VA-E^{(1)}V)JW\xi_\rho\Vert^2\sigma(B)^2+
 \Vert V(A-\rho(A))JW\xi_\rho\Vert^2\eta(B)^2 \nonumber\\
 &\;\hspace{3mm}+2\Vert (VA-E^{(1)}V)JW\xi_\rho\Vert \eta(B)
\sqrt{\Vert V(A-\rho(A))JW\xi_\rho\Vert^2 \sigma(B)^2-D_{AB}^2}\geq D_{AB}^2,
\end{align}
which completes the proof of Theorem \ref{EDUR}.
\end{proof}

\section{Uncertainty Relation for Simultaneous Measurement}\label{5}
We shall define errors of $A$ and $B$ for simultaneous measurements of $A$ and $B$ in this section.
Let $\mathcal{I}$ be a CP instrument for $(\mathcal{M},\mathbb{R}^2)$
and $(\mathcal{K},\pi,E,V)$ the minimal dilation of $\mathcal{I}$.
We define spectral measures $E_{\bf x}$, $E_{\bf y}$ on $\mathcal{K}$ by
\begin{align}
E_{\bf x}(\Delta) &= E(\Delta\times\mathbb{R}),\hspace{5mm}\Delta\in\mathcal{B}(\mathbb{R}),\\
E_{\bf y}(\Gamma) &= E(\mathbb{R}\times\Gamma),\hspace{5mm}\Gamma\in\mathcal{B}(\mathbb{R}),
\end{align}
respectively.
It is natural (owing to the discussion in Section \ref{2}) to consider
that $\mathcal{I}$ corresponds to a measuring apparatus with two output variables ${\bf x}$ and ${\bf y}$,
both of which take values in $(\mathbb{R},\mathcal{B}(\mathbb{R}))$.
We then assume that we use ${\bf x}$ and ${\bf y}$ to measure $A$ and $B$, respectively.
Moreover, we assume that $E_{\bf x}^{(1)}$ and $E_{\bf y}^{(1)}$ are bounded operators on
$\mathcal{K}$ for simplicity, so that
$E_{\bf x}^{(2)}=(E_{\bf x}^{(1)})^2$ and $E_{\bf y}^{(2)}=(E_{\bf y}^{(1)})^2$.
We can naturally define errors of $A$ and $B$ in terms of $\mathcal{I}$ by
\begin{align}
\varepsilon(A,\rho;\mathcal{I}) &=
\langle \rho,V^\ast E_{\bf x}^{(2)}V-AV^\ast E_{\bf x}^{(1)}V
-V^\ast E_{\bf x}^{(1)}VA+A^2 \rangle^{\frac{1}{2}}\\
&=\langle \rho,(E^{(1)}_{\bf x}V-VA)^\ast (E^{(1)}_{\bf x}V-VA)\rangle^{\frac{1}{2}},\\
\varepsilon(B,\rho;\mathcal{I}) &=
\langle \rho, V^\ast E_{\bf y}^{(2)}V-BV^\ast E_{\bf y}^{(1)}V
-V^\ast E_{\bf y}^{(1)}VB+B^2\rangle^{\frac{1}{2}}\\
&=\langle \rho,(E^{(1)}_{\bf y}V-VB)^\ast (E^{(1)}_{\bf y}V-VB) \rangle^{\frac{1}{2}},
\end{align}
respectively.
If $\mathcal{I}$ is realized by a measuring process $(\mathcal{K},\sigma,F,U)$ for $(\mathcal{M},\mathbb{R}^2)$,
we have
\begin{align}
\varepsilon(A,\rho;\mathcal{I}) &=\langle 
\tilde{\rho}\otimes\sigma,(U^\ast(1\otimes F^{(1)}_{\bf x})U-A\otimes1)^2 \rangle^{\frac{1}{2}},\\
\varepsilon(B,\rho;\mathcal{I}) &=\langle 
\tilde{\rho}\otimes\sigma,(U^\ast(1\otimes F^{(1)}_{\bf y})U-B\otimes1)^2 \rangle^{\frac{1}{2}},
\end{align}
respectively,
where $\tilde{\rho}$ is a normal state on $\textrm{\boldmath $B$}(\mathcal{H})$
such that $\rho(M)=\tilde{\rho}(M)$ for all $M\in\mathcal{M}$,
and $F_{\bf x}$ and $F_{\bf y}$ are spectral measures on $\mathcal{K}$ defined by
$F_{\bf x}(\Delta)= F(\Delta\times\mathbb{R})$ for all $\Delta\in\mathcal{B}(\mathbb{R})$
and $F_{\bf y}(\Gamma) = F(\mathbb{R}\times\Gamma)$ for all $\Gamma\in\mathcal{B}(\mathbb{R})$,
respectively.

We have the following uncertainty relation for simultaneous measurements of two different observables.
Let $(\mathcal{M},\mathcal{H},\mathcal{P},J)$ be a standard form
with a $\sigma$-finite von Neumann algebra $\mathcal{M}$.
\begin{theorem}[]\label{EEUR}
Let $A,B$ be elements of $\mathcal{M}_{s.a.}$, $\rho\in\mathcal{S}_n(\mathcal{M})$,
$\mathcal{I}$ a CP instrument for $(\mathcal{M},\mathbb{R}^2)$ and
$(\mathcal{K},E,\pi,V)$ the minimal dilation of $\mathcal{I}$.
Assume that $E_{\bf x}^{(1)}$ and $E_{\bf y}^{(1)}$ are bounded operators on $\mathcal{K}$. Then we have
\begin{equation}
\varepsilon(A)^2\sigma(B)^2+\sigma(A)^2\varepsilon(B)^2+2\varepsilon(A)\varepsilon(B)
\sqrt{\sigma(A)^2\sigma(B)^2-D_{AB}^2}
\geq D_{AB}^2.
\end{equation}
\end{theorem}

\begin{proof}
We define a real inner product $\langle \cdot|\cdot \rangle_\mathbb{R}$ on $\mathcal{K}$ by
\begin{equation}
\langle x|y \rangle_\mathbb{R} =\mathrm{Re}\langle x|y \rangle
\end{equation}
for all $x,y\in\mathcal{K}$,
and put
\begin{align}
a &= V(A-\rho(A))JW\xi_\rho,\\
b &= -iV(B-\rho(B))J\xi_\rho,\\
m &= (E_{\bf x}^{(1)}-\rho(A))VJW\xi_\rho,\\
n &= -i(E_{\bf y}^{(1)}-\rho(B))VJ\xi_\rho.
\end{align}
Then we have
$\langle m|n \rangle_\mathbb{R} = 0$,
$\langle a|b \rangle_\mathbb{R} =D_{AB}$, 
$\Vert a\Vert_\mathbb{R}^2 \leq \sigma(A)^2$,
$\Vert a-m\Vert_\mathbb{R}^2 \leq \varepsilon(A)^2$,
$\Vert b\Vert_\mathbb{R}^2 = \sigma(B)^2$ and
$\Vert b-n\Vert_\mathbb{R}^2=\varepsilon(B)^2$.
By Eq. (\ref{Branciard}), we get the desired inequality.
\end{proof}

\section{Concluding Remarks}\label{6}
Recently, Busch, Lahti, and Werner \cite{BLW14RMP} raised 
a reliability problem for quantum generalizations of the classical 
root-mean-square (rms) error, comparing the noise-operator based error, 
which we adopted here, 
with the Wasserstein 2-distance, another error measure 
based on the distance between probability measures.
They pointed out several discrepancies between those two error measures
in favor of the latter, and claimed that a state-dependent formulation 
for measurement uncertainty relations is not tenable. 

In order to resolve the conflict, one of the authors \cite{18PPT}
introduced the following requirements for 
any sensible error measure generalizing the classical rms error:

(I) The operational definability: The error measure be defined by the POVM of the measurement,
the observable to be measured, and the state of the system to be measured.

(II) The correspondence principle: The error measure coincide with the classical rms
error in the case where there exists the joint probability distribution 
for the observable to be measured
just before the measurement and the meter observable just after the measurement.

(III) The soundness: The error measure take the value 0 if the measurement 
is precise in the sense that the observable to be measured 
just before the measurement and the meter observable just after the measurement are 
perfectly correlated \cite{05PCN,06QPC}.

It was shown that the noise-operator based error satisfies all the requirements, (I)--(III),
whereas the Wasserstein 2-distance does not satisfy (II).  
Thus, the Busch-Lahti-Werner criticism based on the comparison with
the Wasserstein 2-distance is not relevant, and their opinion against 
the state-dependent formulation is unfounded.   
In fact, in the classical case, where the observable algebra $\mathcal{M}$ is abelian,  
the noise-operator-based error takes the value 0 if and only if 
the measured observable just before the measurement and the meter observable just after the measurement take the same value with probability 1, 
whereas the Wasserstein 2-distance takes the value 0 if and only if they only have 
the same probability distribution.  For more detailed discussions, we refer the reader to 
Ref.~\cite{18PPT}. 

In contrast to the violation of Heisenberg's relation \eq{Hei27} for the noise-operator-based error $\ep$, a state-dependent error measure, 
Appleby \cite{App98c} considered the state-independent error $\ep_{\rm Appleby}$ 
defined by
\begin{equation}
\ep_{{\rm Appleby}}(A;\mathcal{I})=\sup_{\rho}\ep(A,\rho;\mathcal{I}),
\end{equation}
where the supremum is taken over all pure states $\rho$,
and derived the relation
\begin{equation}
\ep_{{\rm Appleby}}(Q;\mathcal{I})\,\ep_{{\rm Appleby}}(P;\mathcal{I})\ge\frac{\hbar}{2}
\end{equation}
for canonically conjugate observables $Q,P$ and for any joint measurement 
$\mathcal{I}$ of $Q,P$,
except for the case where $\ep_{{\rm Appleby}}(Q;\mathcal{I})= 0$ or 
$\ep_{{\rm Appleby}}(P;\mathcal{I})=0$.
It should be noted that in the state-independent formulation as above 
the error measures
$\ep_{\rm Appleby}(Q;\mathcal{I})$ and $\ep_{\rm Appleby}(P;\mathcal{I})$
often diverge. In fact, they diverges for almost every linear measurement
 \cite{13DHE}.
Even in the original $\gamma$-ray thought experiment, the error measure
$\ep_{\rm Appleby}(Q;\mathcal{I})$ obviously diverges.
The notion of the resolution power of a microscope is well-defined only in the case
where the object is well-localized in the scope of the microscope, and it cannot be
captured by the state-independent formulation.
Recently, Busch-Lahti-Werner \cite{BLW13,BLW14JMP} revived the state-independent 
formulation similar to Appleby's \cite{App98c}, but the same criticisms 
apply to their approach. In fact, the Busch-Lahti-Werner formulation in \cite{BLW13}
is equivalent to Appleby's formulation \cite{App98c}
for any linear measurements \cite{13DHE}.
For detailed discussions, we refer the reader to Ref.~\cite{13DHE}. 

\section*{Acknowledgment}
The authors acknowledge the supports of the JSPS KAKENHI, No.~26247016, No.~17K19970,  and of the IRI-NU collaboration.
KO is supported by the JSPS KAKENHI, No. 16K17641, No. 17H01277, 
by Research Origin for Dressed Photon,
by Grant for Basic Science Research Projects from The Sumitomo Foundation
and by the JST CREST, Grant Number JPMJCR17N2, Japan.

\end{document}